\documentclass[runningheads]{llncs}
\pdfoutput=1
\usepackage{latexsym}
\usepackage{amsmath}
\usepackage{amssymb}
\usepackage{graphicx}
\usepackage{xspace}
\usepackage{todonotes}
\usepackage[font=small, labelfont=bf]{caption}
\usepackage[subrefformat=parens, labelfont=default]{subfig}
\usepackage{times}
\usepackage{enumerate}

\newcommand{\dc}{\mathrm{DC}}

\newcommand{\NPH}{\textnormal{\texttt{NP}}-hard\xspace}
\newcommand{\NPC}{\textnormal{\texttt{NP}}-complete\xspace}

\begin{document}

\title{On Maximum Differential Coloring of Planar Graphs}

\author{M.~A.~Bekos\inst{1}, M.~Kaufmann\inst{1}, S.~Kobourov\inst{2}, S.~Veeramoni\inst{2}}
%
\authorrunning{M.~A.~Bekos, M.~Kaufmann, S.~Kobourov, S.~Veeramoni}
%
\tocauthor{M.~A.~Bekos, M.~Kaufmann, S.~Kobourov, S.~Veeramoni}
%
\authorrunning{M.~A.~Bekos, M.~Kaufmann, S.~Kobourov, S.~Veeramoni}

\institute{%
    Wilhelm-Schickard-Institut f\"{u}r Informatik - Universit\"{a}t T\"{u}bingen, Germany\\
    \and
    Department of Computer Science - University of Arizona, Tucson AZ, USA\\
}

\maketitle

\begin{abstract}
We study the \emph{maximum differential coloring problem}, where the
vertices of an $n$-vertex graph must be labeled with distinct
numbers ranging from $1$ to $n$, so that the minimum absolute
difference between two labels of any two adjacent vertices is
maximized. As the problem is \NPH for general
graphs~\cite{leung1984}, we consider planar graphs and subclasses
thereof. We initially prove that the maximum differential coloring
problem remains \NPH, even for planar graphs. Then, we present tight
bounds for regular caterpillars and spider graphs. Using these new
bounds, we prove that the Miller-Pritikin labeling
scheme~\cite{miller89} for forests is optimal for regular
caterpillars and for spider graphs. Finally, we describe
close-to-optimal differential coloring algorithms for general caterpillars
and biconnected triangle-free outer-planar graphs.
\end{abstract}


\section{Introduction}
\label{sec:intro}

The four color theorem states that only four colors are needed to
color any map, so that no neighboring countries share the same
color. However, if the countries in the map are not all contiguous,
then the result no longer
holds~\cite{Gansner_Hu_Kobourov_2009_gmap}. In order to avoid
ambiguity, this necessitates the use of a unique color for each
country. As a result, the number of colors needed is equal to the
number of countries.

Given a map, define the country graph $G = (V,E)$ to be the graph
where countries are vertices and two countries are connected by an
edge if they share a nontrivial border. In the {\em maximum
differential coloring problem}~\cite{leung1984} the goal is to find
a labeling of the $n$ vertices of graph $G$ with distinct numbers
ranging from $1$ to $n$ (treated as {\em colors}),  which maximizes
the absolute label difference among adjacent vertices. More
formally, let $C = \{c \mid c: V \to \{1, 2, \ldots , |V|\}\}$ be the set
of one-to-one functions for labeling the vertices of $G$. For any
$c\in C$, the {\em differential coloring} achieved by $c$ is
$\min_{(i,j) \in E} |c(i) - c(j)|$. We seek the labeling function $c
\in C$ that achieves the maximum differential coloring: $\dc(G) =
\max_{c\in C} \min_{(i,j) \in E} |c(i) - c(j) |$, which is the
\emph{differential chromatic number} of $G$.

The maximum differential coloring problem is 
in a sense the opposite of the well-studied {\em bandwidth
minimization problem}, which is known to be
\NPC~\cite{Monien86,Papadimitriou_1975}. Optimal algorithms for the
bandwidth minimization problem are known only for restricted classes
of graphs, e.g., caterpillars with hair length
$1$~\cite{NYAS:NYAS344}, caterpillars with hair length
$3$~\cite{assmann1981bandwidth}, chain
graphs~\cite{kloks1998bandwidth}, co-graphs~\cite{yan1997bandwidth},
bipartite permutation graphs~\cite{Heggernes07bandwidthof}, AT-free
graphs~\cite{golovach2009bandwidth}. As in many graph theoretic
maximization vs minimization problems (e.g., shortest vs longest
path), results for bandwidth minimization do not translate into
results for maximum differential coloring. Although the maximum
differential coloring problem is less known than the bandwidth
minimization problem, it has received considerable attention
recently. In addition to
map-coloring~\cite{Gansner_Hu_Kobourov_2009_gmap}, the problem is
motivated by the \emph{radio frequency assignment problem}, where
$n$ transmitters have to be assigned $n$ frequencies, so that
interfering transmitters have frequencies as far apart as
possible~\cite{hale80}.

\subsection{Previous Work}
\label{sec:previouswork}

The maximum differential coloring problem was initially studied  in
the context of multiprocessor scheduling under the name ``separation
number'' by Leung {\em et al.}~\cite{leung1984}, who showed that the
problem is \NPC. Twenty years later, Yixun {\em et
al.}~\cite{yixun2003dual} studied the same problem under the name
``dual-bandwidth'' and gave several upper bounds, including the
following simple bound for connected graphs:

\begin{property}
For any connected graph $G$, $\dc(G) \le \lfloor \frac{n}{2}
\rfloor$~{\em\cite{yixun2003dual}}.
\label{prp:upperbd}
\end{property}

The proof is straightforward: one of the vertices of $G$ has to be
labeled $\lceil \frac{n}{2} \rceil$ and since $G$ is connected that
vertex must have at least one neighbor which (regardless of its label)
would make the difference along that edge at most $\lfloor \frac{n}{2}
\rfloor$.

The maximum differential coloring problem is also known as the
``anti-bandwidth problem''~\cite{abktree}. Heuristics for the
maximum differential coloring problem have been suggested by Duarte
{\em et al.}~\cite{grasp} using LP-formulation, by Bansal {\em et
al.}~\cite{memetic} using memetic algorithms and by Hu {\em et
al.}~\cite{hu2010} using spectral based methods. Another line of
research focuses on solving the maximum differential coloring
problem optimally for special classes of graphs, e.g., Hamming
graphs~\cite{abhamming}, meshes~\cite{3dmesh},
hypercubes~\cite{abrsstv,Wang20091947}, complete binary
trees~\cite{weili} and complete $k$-ary trees for odd values of
$k$~\cite{abktree}. Isaak {\em et al.}~\cite{Isaak98powersof} give a
greedy algorithm for the differential chromatic number of complement
of interval and threshold graphs by computing the  $k$-th power of a
Hamiltonian path. Weili {\em et al.}~\cite{weili} compute the
differential chromatic number for what they call ``caterpillars''
(but which differ from the standard graph-theoretic caterpillars).

Miller and Pritikin~\cite{miller89} describe a labeling scheme
which, for a forest $G$ with bipartition $U$ and $V$, gives a
differential coloring value equal to the size of the smaller vertex set, i.e.,
$min\{|U|,|V|\}$. In high-level description, this approach can be
summarized as follows. Say, without loss of generality, that $|U|
\le |V|$. The vertices in $U$ are labeled with labels from the
``minority interval'' $I_{\mathrm{min}} = [1,|U|]$, while the
vertices in $V$ are labeled with labels from the ``majority
interval'' $I_{\mathrm{maj}} = [|U|+1,|V|]$. Since the average
degree of the vertices in $V$ is $|U| + |V| - 1 / |V| < 2$, there
exists at least one vertex in $V$, say $v$, with degree $\le 1$.
Based on the vertex $v$, a vertex $u \in U$ is chosen as follows: If $deg(v) = 1$, then
$u$ is the neighbor of $v$. Otherwise, $u$ is arbitrarily chosen
from $U$. Both $v$ and $u$ are then labeled with the smallest
available labels from $I_{\mathrm{maj}}$ and $I_{\mathrm{min}}$,
respectively, and removed from $G$. This procedure is repeated until
$U$ is empty. The remaining vertices in $V$ (if any) are labeled
with the remaining available labels in $I_{\mathrm{maj}}$. Note that
when a vertex $u \in U$ is labeled, a vertex $v \in V$ is also
labeled. Hence, as long as $U$ is non-empty, the number of labeled
vertices of $U$ is equal to the number of labeled vertices of $V$.
This implies that the minimum label difference between any two
neighboring vertices in $G$ is at least $|U|$.

A closely related problem to the maximum differential coloring
problem is the \emph{equitable coloring problem}~\cite{hajnal70}.
Formally, an equitable coloring is an assignment of colors to the
vertices of a graph, so that no two adjacent vertices have the same
color and the number of vertices in any two color classes differ by
at most one. The problem of deciding whether a graph admits an
equitable coloring with no more than three colors is
\NPC~\cite{kubale2004graph}. If a graph $G$ has differential
chromatic number $k$, then the vertices labeled $[1,k] , [k+1,2k]
\cdots$ form equitably colored classes and so $G$ has an equitable
coloring with $\lfloor\frac{n}{k}\rfloor + 1$ colors. Lin {\em et
al.}~\cite{linantibandwidth} describe a (sub-optimal) labeling for
connected bipartite graphs with a differential coloring of value
$\lfloor\frac{n}{\Delta}\rfloor$, where $\Delta$ is the max degree,
using the relationship between the anti-bandwidth problem and the
equitable coloring problem.

Another related problem is the {\em channel assignment
problem}~\cite{reed2003recent}, in which each edge has a weight and
the objective is to find a labeling of the vertices, so that the
difference between the labels of the endpoints of each edge is at
least equal to its weight. However, the same label can be used by
multiple vertices.

\subsection{Preliminaries}
\label{sec:preliminaries}

Let $G=(V,E)$ be an undirected graph. We denote the number of
vertices of $G$ by $n$, i.e., $n=|V|$. The degree of vertex $u\in V$
is denoted by $d(u)$. The degree of graph $G$ is then defined as:
$\Delta(G)=\max_{u\in V}d(u)$.

A {\em caterpillar} is a tree in which removing all leaves results
in a path; see Fig.~\ref{fig:caterpillar}. Thus, a caterpillar
consist of a simple path, called the ``spine'', and each spine
vertex is adjacent to a certain number of leaves, called the
``legs''. In caterpillars, $\Delta$ refers to the maximum number of
legs of any spine vertex. In a {\em regular caterpillar}, every
spine vertex has the same number of legs. A {\em spider} is a graph
with a center vertex connected to a particular number of disjoint
paths; see Fig.~\ref{fig:spider}. The vertices of a spider have {\em
levels}, according to their distance from the center. In a spider,
$N_e$, $N_o$ and $N_l$ denote the number of even level, number of
odd level and number of vertices in level $l$, respectively. A
\emph{radius-k star} graph is a spider with all paths of the same
length k; see Fig.~\ref{fig:star}.

\begin{figure}[htb]
    \centering
    \begin{minipage}[b]{.48\textwidth}
        \centering
        \subfloat[\label{fig:caterpillar}{A caterpillar}]
        {\includegraphics[width=\textwidth,page=1]{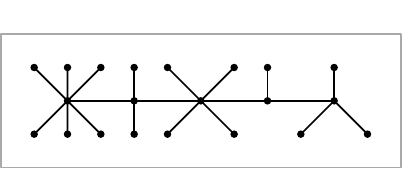}}
    \end{minipage}
    \begin{minipage}[b]{.24\textwidth}
        \centering
        \subfloat[\label{fig:spider}{A spider}]
        {\includegraphics[width=\textwidth,page=2]{caterpillar-spider-star}}
    \end{minipage}
    \begin{minipage}[b]{.24\textwidth}
        \centering
        \subfloat[\label{fig:star}{A radius-3 star}]
        {\includegraphics[width=\textwidth,page=3]{caterpillar-spider-star}}
    \end{minipage}
    \caption{Illustration of (a)~a caterpillar, (b)~a spider (c)~a radius-3 star.}
    \label{fig:caterpillar-spider-star}
\end{figure}

\subsection{Paper Structure}
\label{sec:paperstructure}

This paper is structured as follow: In Section~\ref{sec:np}, we
prove that the differential coloring problem is \NPH even for planar
graphs (Theorem~\ref{thm:np-proof}). In
Section~\ref{sec:opt-reg-cat-spi}, we present tight upper bounds for
regular caterpillars (Theorem~\ref{thm:opt-reg-cat}) and spiders
(Theorem~\ref{thm:opt-spider}). In Section~\ref{sec:regcat-spiders},
we present closed-form optimal labeling schemes (more intuitive than
the known labeling scheme~\cite{miller89}) for regular caterpillars
(Theorem~\ref{thm:reg-cat}) and for spiders with path lengths all
even or all odd (Theorem~\ref{thm:spider}). In
Sections~\ref{sec:gen-cat} and \ref{sec:outerplanar}, we describe
labeling algorithms which produce close-to-optimal labeling for
caterpillars (Theorem~\ref{thm:gen-cat}) and biconnected
triangle-free outer-planar graphs (Theorem~\ref{thm:biconn}),
respectively. We conclude in Section~\ref{sec:conclusion} with open
problems and future work. 
\section{Differential Coloring is NP-complete for Planar Graphs}
\label{sec:np}

In this section, we prove that the differential coloring problem is
\NPH even for planar graphs. 

\begin{theorem}
Given a planar graph $G=(V,E)$ 
it is \NPH to determine the differential chromatic
number of $G$. \label{thm:np-proof}
\end{theorem}

\begin{proof}
In order to prove that the problem is \NPH, we employ a
reduction from the well-known $3$-coloring problem, which is \NPC
for planar graphs 
~\cite{Garey:1979:CIG:578533}.

More precisely, let $G=(V,E)$ be an instance of the $3$-coloring
problem mentioned above, i.e., graph $G$ is an $n$-vertex planar
graph. In the following,
we will construct a new planar graph $G'$, so that $G'$ has
differential coloring of value at least $n$ if and only if $G$ is
$3$-colorable.

Graph $G'=(V',E')$ is constructed by attaching a path $v \rightarrow
v_1 \rightarrow v_2$ to each vertex $v \in V$ of $G$; see
Figs.~\ref{fig:instance} and \ref{fig:cunstruction}. Hence, we can
assume that $V'=V \cup V_1 \cup V_2$, where $V$ is the vertex set of
$G$, $V_1$ contains the first vertices of each 2-length path and
$V_2$ the second ones. Clearly, $G'$ is planar on $3n$ vertices.
Now, observe that if $G'$ is $3$-colorable, then $G$ is
$3$-colorable, as well. This is because $G$ is a subgraph of $G'$.
On the other hand, if $G$ is $3$-colorable, then $G'$ is also
$3$-colorable; for each vertex $v \in V$, simply color its neighbors
$v_1$ and $v_2$ with two distinct colors different from the color of
$v$. Next, we show that $G'$ is $3$-colorable if and only if $G'$
has differential coloring of value at least $n$.

\begin{figure}[t]
    \centering
    \begin{minipage}[b]{.32\textwidth}
        \centering
        \subfloat[\label{fig:instance}{Instance $G=(V,E)$}]
        {\includegraphics[width=\textwidth,page=1]{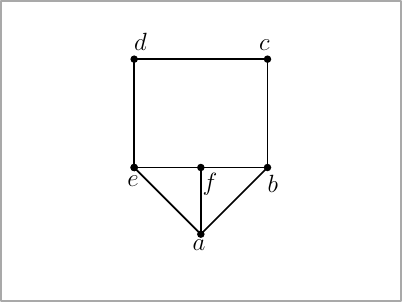}}
    \end{minipage}
    \begin{minipage}[b]{.32\textwidth}
        \centering
        \subfloat[\label{fig:cunstruction}{Graph $G'=(V',E')$}]
        {\includegraphics[width=\textwidth,page=2]{reduction}}
    \end{minipage}
    \begin{minipage}[b]{.32\textwidth}
        \centering
        \subfloat[\label{fig:labeling}{Differential labeling}]
        {\includegraphics[width=\textwidth,page=3]{reduction}}
    \end{minipage}
    \caption{(a)~An instance $G=(V,E)$ of the $3$-coloring problem;
             (b)~An instance $G'=(V',E')$ of the maximum differential coloring problem constructed based on graph $G$;
             (c)~The differential labeling of $G'$, in the case where $G$ is $3$-colorable.}
    \label{fig:reduction}
\end{figure}

First assume that $G'$ has differential coloring of value at least
$n$ and let $l: V' \to \{1, 2, \ldots , 3n\}$ be the respective
labeling. We proceed to color the vertices of $G'$. Let $u \in V'$
be a vertex of $G'$. We assign a color $c(u)$ to $u$ as follows:

\begin{enumerate}[-]
  \item If $1 \leq l(u) \leq n$, then $c(u)=1$
  \item If $n+1 \leq l(u) \leq 2n$, then $c(u)=2$
  \item If $2n+1 \leq l(u) \leq 3n$, then $c(u)=3$
\end{enumerate}

Since labeling $l$ guarantees a differential coloring of value at
least $n$, no two vertices with the same color are adjacent. Hence,
coloring $c$ is a $3$-coloring for $G'$.

Now, consider the case where $G'$ is $3$-colorable. Let $C_i
\subseteq V$ be the set of vertices of the input graph $G$  with
color $i$, $i=1,2,3$. Obviously, $C_1 \cup C_2 \cup C_3 = V$. We
proceed to compute a labeling $l$ of the vertices of graph $G'$ as
follows (see Fig.~\ref{fig:labeling}):

\begin{enumerate}[-]
  \item Vertices in $C_1$ are labeled with labels from $1$ to $|C_1|$.
  \item Vertices in $C_2$ are labeled with labels from $n+|C_1|+1$ to $n+|C_1|+|C_2|$.
  \item Vertices in $C_3$ are labeled with labels from $2n+|C_1|+|C_2|+1$ to $2n+|C_1|+|C_2|+|C_3|$.
  \item For a vertex $v_1 \in V_1$ neighboring to a vertex $v \in C_1$, $l(v_1) = l(v) + n$.
  \item For a vertex $v_1 \in V_1$ neighboring to a vertex $v \in C_2$, $l(v_1) = l(v) - n$.
  \item For a vertex $v_1 \in V_1$ neighboring to a vertex $v \in C_3$, $l(v_1) = l(v) - 2n$.
  \item For a vertex $v_2 \in V_2$ neighboring to a vertex $v_1 \in V_1$, $l(v_2) = l(v_1)+n +|C_2|$.
\end{enumerate}

From the above it follows that the label difference between (i)~any
two vertices in $G$, (ii)~a vertex $v_1 \in V_1$ and its neighbor $v
\in V$, and, (iii)~a vertex $v_1 \in V_1$ and its neighbor $v_2 \in
V_2$ is at least $n$. So, $G'$ has differential coloring of value at
least $n$.
\end{proof}

\section{Upper Bounds for Regular Caterpillars and Spiders}
\label{sec:opt-reg-cat-spi}

In this section, we establish new upper bounds for $\dc(G)$, when
$G$ is a regular caterpillar or a spider.  Then, we show that the
Miller-Pritikin labeling scheme is optimal for these classes of
graphs.

\begin{theorem}
Let $G$ be a $\Delta$-regular caterpillar with $n$ vertices. If $G$
has an odd number of spine vertices, then $\dc(G) \le \lceil
\frac{n-\Delta}{2} \rceil$. Otherwise $\dc(G) \le \lfloor
\frac{n}{2} \rfloor$.
\label{thm:opt-reg-cat}
\end{theorem}

\begin{proof}
If $G$ has an even number of spine vertices, then by
Property~\ref{prp:upperbd} it follows that $\dc(G) \le \lfloor
\frac{n}{2} \rfloor$. We will show that, when $G$ has an odd number
of spine vertices, $\dc(G) \le \lceil\frac{n-\Delta}{2}\rceil$. Let
the number of spine vertices be $s = 2k +1$, for some $k \geq 1$.
Then, the total number of vertices of $G$ is $n = (2k+1)(\Delta
+1)$. Note that $\lceil\frac{n-\Delta}{2}\rceil = 1 + k(\Delta +1)$.
For a proof by contradiction assume that there exists a labeling of
value $c^* = \lceil\frac{n-\Delta}{2}\rceil + 1$.

\begin{proof}
\begin{lemma}
No spine vertex is labeled in the interval
$[\lceil\frac{n-\Delta}{2}\rceil,\lceil\frac{n + \Delta}{2}\rceil]$
\label{lem:spines}
\end{lemma}
\begin{proof}
Assume to the contrary that $ i \in
 [\lceil\frac{n-\Delta}{2}\rceil,\lceil\frac{n + \Delta}{2}\rceil]$
is a spine vertex label. Consider the labels that can be assigned to
the $\Delta$ legs of $i$ (with a slight abuse of notation $i$ also
refers to the vertex labeled $i$). To achieve value $c^*$ the label
for a leg of $i$ can either lie in the interval $L = [1, i-
(\lceil\frac{n-\Delta}{2}\rceil + 1)]$ or in the interval $H= [i +
\lceil\frac{n-\Delta}{2}\rceil + 1,n]$. We consider three cases for
the label of $i$.
\begin{enumerate}
\item[Case 1:] $ i \ne \lceil\frac{n-\Delta}{2}\rceil$ and $ i \ne
\lceil\frac{n+\Delta}{2}\rceil$. In this case, the total number of
labels in $L$ and $H$ is:
\begin{eqnarray*}
i - \left(\left\lceil\frac{n-\Delta}{2}\right\rceil + 1\right) + n - \left(i + \left\lceil\frac{n-\Delta}{2}\right\rceil \right) & = &  n - \left(2 \left\lceil\frac{n-\Delta}{2}\right\rceil  + 1 \right) \\
& = &  n - (2k (\Delta + 1) + 3 ) \\
& = &  n - ( n + 2 - \Delta) ~ = ~ \Delta - 2
\end{eqnarray*}

\item[Case 2:] $i = \lceil\frac{n-\Delta}{2}\rceil$. In this case,
$L$ is empty and all leg labels lie in interval $H$. Hence, the
total number of labels in $H$ is:
\begin{eqnarray*}
n - \left(2 \cdot \left\lceil\frac{n-\Delta}{2}\right\rceil\right) & = &  n - (2 \cdot ((k \cdot (\Delta + 1)) + 1) \\
& = &  ((2k+1)\cdot(\Delta+1)) - (2 \cdot ((k \cdot (\Delta + 1)) + 1))\\
& = &  \Delta -1
\end{eqnarray*}

\item[Case 3:] $ i = \lceil\frac{n+\Delta}{2}\rceil$. In this case,
$H$ is empty and all leg labels lie in interval $L$. Similarly to
the previous case, the total number of labels in $L$ is: $i -
\left(\left\lceil\frac{n-\Delta}{2}\right\rceil + 1\right) = \Delta
- 1$
\end{enumerate}
\noindent In all cases the labels for the legs of $i$ are
insufficient. So, we have a contradiction.
\end{proof}
\end{proof}

Back to the theorem: By Lemma~\ref{lem:spines}, it follows that
labels of spine vertices either lie in the interval $L_s =
[1,\lceil\frac{n-\Delta}{2}\rceil -1]$ or in the interval $H_s =
[\lceil\frac{n+\Delta}{2}\rceil +1 ,n]$. Observe that the maximum
difference between any two elements in the interval $L_s$ is
$\lceil\frac{n-\Delta}{2}\rceil - 2$. This suggests that in order to
achieve differential coloring $c^*$, adjacent spine vertices cannot
both be labeled from the interval $L_s$. Similarly, we can prove
that adjacent spine vertices cannot both be labeled from the
interval $H_s$, as the maximum difference between two elements in
$H_s$ is $n - (\lceil\frac{n+\Delta}{2}\rceil + 1) = n -
(\lceil\frac{n-\Delta}{2}\rceil + \Delta +1) =  n - (k \cdot (\Delta
+ 1) + 1 +  \Delta +1) = \lceil\frac{n-\Delta}{2}\rceil  < c^*$.

From the above it follows that the labels for spine vertices must
alternate between interval $L_s$ and $H_s$ such that for the labels
of the $2k + 1$ spine vertices, one of the intervals supplied $k+1$
labels and other interval supplied $k$ labels. Assume without loss
of generality that $L_s$ supplies $k+1$ labels. In order to achieve
differential coloring $c^*$, the $(k+1)\Delta$ legs of these spine
vertices must all have labels in the interval $I =
[\lceil\frac{n-\Delta}{2}\rceil + 2,n]$. As $\Delta \ge 1$, interval
$I \supseteq H_s$, and so $I$ must also contain the $k$ labels $H_s$
supplies for spine vertices. Thus, in total $I$ must contain at
least $(k+1)\Delta + k$ labels. However, the size of the interval
$I$ is:
\begin{eqnarray*}
1 +  n - (\lceil\frac{n-\Delta}{2}\rceil +2) & = &  1 + ((2k+1)\cdot(\Delta+1)) - ((k \cdot (\Delta + 1)) + 1) - 2 \\
& = &  k \cdot\Delta + k + \Delta -1\\
& < &  (k+1)\Delta + k
\end{eqnarray*}
So, we have a contradiction.
\end{proof}

\begin{corollary}
The Miller-Pritikin labeling scheme is optimal for regular
caterpillars.
\label{cor:opt-miller-reg-cat}
\end{corollary}
\begin{proof}
Let $G$ be a regular caterpillar on $n$ vertices. First, consider
the case where $G$ has an even number of spine vertices, say $s =
2k$ for some $k \geq 1$. $G$ is a bipartite graph
whose vertices form two disjoint sets $U$ and $V$, where $U$
consists of the $k$ odd spine vertices and the $k\Delta$ legs of the
even spine vertices, and, $V$ consists of the $k$ even spine
vertices and the $k\Delta$ legs of the odd spine vertices. So, $ |U|
=  |V| = k + k\Delta = \frac{n}{2}$. Since the Miller-Pritikin
labeling scheme yields a labeling with value equal to the size of
the smaller vertex set, the labeling is optimal by
Property~\ref{prp:upperbd}.

Now, consider the case where $G$ has an odd number of spine
vertices, say $s =2k+1$ for some $k \geq 1$. In this case, $U$
consists of the $k$ even spine vertices and the $(k+1)\Delta$ legs
of the odd spine vertices, and, $V$ consists of $k+1$ odd spine
vertices and the $k\Delta$ legs of the even spine vertices.  So, $
|U| = k + (k+1)\Delta$ and $ |V| = k + 1 + k\Delta$, and $\min\{
|U|, |V|\} = k + 1 + k\Delta = \lceil\frac{n-\Delta}{2}\rceil$.
Thus, by Theorem~\ref{thm:opt-reg-cat} the Miller-Pritikin labeling
scheme is optimal.
\end{proof}

In the following, we present a tight upper bound for spider graphs.
However, before presenting our labeling method, we make a few simple
observations about spider graphs. Let $p$ be the number of paths
connected to the center vertex $v_c$ in a spider graph $G$. Recall
that by $N_e$, $N_o$ and $N_l$ we denote the number of even level,
number of odd level and number of vertices in level $l$,
respectively. Then, the number of vertices of $G$ is:

\begin{equation}
n = N_{e} + N_{o} + 1
\label{eq:nspider}
\end{equation}

Each of the $p$ paths of $G$ starts with an odd level vertex and
alternates between even and odd levels. It follows that on each path
the number of odd level vertices is at most one more than the
even level vertices. Summing over all $p$ paths we get:

\begin{equation}
N_{o} - N_{e} \le p
\label{eq:odd-even}
\end{equation}

\begin{theorem} \label{thm:opt-spider}
If $G$ is a spider graph with $N_e$ even level vertices, then
$\dc(G) \le N_{e} + 1$.
\end{theorem}
\begin{proof}
For a proof by contradiction suppose that there exists a labeling of
value  $c^* = N_{e} + 2$.

\begin{proof}
\begin{lemma}
The center vertex label is not in the interval $[N_{e} + 1, N_{e} + p + 1]$.
\label{lem:center}
\end{lemma}
\begin{proof}
For the sake of contradiction, let $ i \in [N_e + 1, N_e + p + 1]$
be the label of the center vertex and consider the labels that can
be assigned to the $p$ vertices of level $1$. To achieve a
differential coloring of value $c^*$, the labels of the level-$1$
vertices can either lie in the interval $L = [1, i- (N_e+2)]$ or in
the interval $H= [N_e + 2 +i, n]$. We consider three cases for the
values of $i$.
\begin{enumerate}
\item[Case 1:] $ i \ne N_{e} + 1$ and  $ i \ne N_{e} + p + 1$.
Then, by Equations~\ref{eq:nspider} and \ref{eq:odd-even} the total
number of labels in $L$ and $H$ is:

\begin{eqnarray*}
i - N_{e} - 2 + n -( i+ N_{e} + 2 ) + 1 & = &  n - 2 N_{e} -3  \\
& = &  N_{e} + N_{o} + 1 - 2 N_{e} -3 \quad\\
& = &  N_{o} - N_{e} -2 \le p -2 \quad
\end{eqnarray*}

\item[Case 2:] $ i = N_{e} + 1$. In this case, $L$ is empty and all
labels of level-$1$ vertices lie in $H$. By
Equations~\ref{eq:nspider} and \ref{eq:odd-even}, the total number
of labels in $H$ is:

\begin{eqnarray*}
n -( i+ N_{e} + 2 ) + 1 & = & n -(N_{e} + 1+ N_{e} + 2 ) + 1 \\
& = & N_{e} + N_{o} + 1-(N_{e} + 1+ N_{e} + 2 ) + 1 \quad \\
& = & N_{o} - N_{e} -1 \le p -1  \quad
\end{eqnarray*}

\item[Case 3:] $ i = N_{e} + p + 1$. In this case, $H$ is empty and
all labels of level-$1$ vertices lie in $L$. Similarly to the
previous case, the total number of labels in $L$ is: $ i - N_{e} - 2
=  N_{e} + p + 1 - N_{e} - 2 =  p -1 $

\end{enumerate}

\noindent In all cases the number of labels is less than $p$,
arriving to a contradiction and completing the proof of this lemma.
\end{proof}
\end{proof}

Back to the theorem: By Lemma~\ref{lem:center}, the center label
either lies in interval $L_c = [1,N_{e}]$ or in interval $H_c =
[N_{e} + p + 2 ,n]$. Let us first assume that the center label lies
in $L_c$ . In this case, the level-$1$ vertices should lie in the
interval $I = [N_{e} + 3 , n]$. Note that in order to achieve a
differential coloring of value $c^*=N_{e} + 2$, adjacent vertices
from neighboring levels $2j$ and $2j+1$ cannot both lie in the
interval $[1,N_{e} + 2]$. Also, $N_{2j+1} \le N_{2j}$. It follows
that the labels of at least $N_{2j+1}$ vertices lie in the interval
$I$. So, interval $I$ must contain at least $N_1 + N_3 + \cdots  =
N_{o}$ elements. The contradiction follows from the size of the
interval $I$, which is:

\begin{eqnarray*}
n - N_{e} - 3  + 1 & = & N_{e} + N_{o} + 1  - N_{e} - 3  + 1 \quad \\
& = & N_{o} -1  \quad
\end{eqnarray*}

Now, assume that the center lies in the interval $H_c$. An analogous
argument shows that the interval $I' = [1,n-N_{e}-2]$ must contain
at least $N_{o}$ elements which is more than the size of $I'$,
leading to a contradiction. As both cases result in contradictions,
this completes the proof of Theorem~\ref{thm:opt-spider}.
\end{proof}

\begin{corollary}
The Miller-Pritikin labeling scheme is optimal for spiders.
\label{cor:opt-miller-spider}
\end{corollary}
\begin{proof}
Let $G$ be a spider graph. Clearly, $G$ is a bipartite graph whose
vertices form disjoint sets $U$ and $V$, where the even level
vertices and the center vertex form $U$ and the odd level vertices
form $V$. Labeling $G$ with the Miller-Pritikin scheme gives a
differential coloring of value at least $m= \min \{ |U|, |V| \}=
\min \{ N_e +1, N_o \}$.

We now prove that $m$ is optimal. We have that $N_e \le N_o$. If
$N_e = N_o$, then $m=N_o$. By Property~\ref{prp:upperbd}, it follows
that $\dc(G)$ is at most $\lfloor n/2 \rfloor$, which by
Equation~\ref{eq:nspider} is at most $\lfloor (N_o + N_e + 1)/2
\rfloor = \lfloor (2N_o + 1)/2 \rfloor = N_o = m$. Now, assume that
$N_e < N_o$. In this case, $m = N_e + 1$ which is optimal by
Theorem~\ref{thm:opt-spider}, completing the proof.
\end{proof}

\section{Optimal labeling for regular caterpillars and spiders with path lengths all even or all odd}
\label{sec:regcat-spiders}

In this section, we describe two optimal labeling schemes for
regular caterpillars and spiders with path lengths all even or all
odd, respectively. Note that by
Corollaries~\ref{cor:opt-miller-reg-cat} and
\ref{cor:opt-miller-spider} the Miller-Pritikin labeling scheme is
also optimal for these classes of graphs. However, the labeling
schemes that we present in this section are more intuitive, more
structured and therefore of a simple nature compared to the
Miller-Pritikin labeling scheme.

\subsection{Optimal labeling for regular caterpillars}
\label{sec:regcat}

Let $G$ be an $n$-vertex regular caterpillar in which each spine
vertex has $\Delta \ge 1$ legs. Let $s$ denote the number of spine
vertices. Then, as $n = s \cdot (\Delta+1)$, we have $s =
\frac{n}{\Delta+1}$.

It is always good to label all legs of a spine vertex $v$ from an
interval of consecutive numbers, since the maximum difference
between a spine vertex $v$ and its legs depends only on the
difference between the label of $v$ and the highest or lowest label
of the legs of $v$.

First, consider the case that there is an even number of spine
vertices, say $s = 2k$ for some $k \geq 1$. We label the spine
vertices using the $k$ lowest and $k$ highest numbers in an
alternating fashion. Starting with the leftmost spine vertex and
moving to the right, we label the spine vertices as $1,  n - k +1,
2, n - k +2, \ldots $ and so on, ending at the rightmost spine
vertices with numbers $k$ and $n$; see Fig.~\ref{sfg:regcat-even}.
There are only two values for the differences between adjacent spine
vertices, namely $n - k$ and $n - k -1$. As $\Delta \ge 1$, the
difference is at least $n-k-1 = 2k\Delta + k -1 \ge k(\Delta +1) =
n/2$. We denote by $L_s$ the set of spine vertices with labels from
$1$ to $k$, and, by $H_s$ the spine vertices with labels from
$n-k+1$ to $n$.

Next, we split the middle range $[k+1,n-k]$ into two ranges
$L_{\ell} = [k+1,\frac{n}{2}]$ and $H_{\ell} = [\frac{n}{2}+1,n-k]$.
We label the legs of $L_s$ from the range $H_{\ell}$ and the legs of
$H_s$ from the range $L_{\ell}$ as follows. For a spine vertex from
$L_s$ with label $1 \le i \le k$, we label its $\Delta$ legs with
numbers from the interval $[\frac{n}{2}+((i-1) \cdot \Delta) + 1,
\frac{n}{2}+ i \cdot \Delta]$. For a spine vertex from $H_s$ with
label $j = n - k + i$ between $n - k +1 $ and $n$, we label its
$\Delta$ legs with numbers from the interval $[k+((i-1) \cdot
\Delta) + 1, k + i \cdot \Delta]$. It follows that the difference
between a low spine vertex from $L_s$ and one of its legs is at
least:

\begin{eqnarray*}
\frac{n}{2} + (i-1) \cdot \Delta +1 - i & = & n/2 +i(\Delta - 1) - \Delta +1  \quad \\
& = & n/2 +(i-1)(\Delta - 1) \quad
\end{eqnarray*}

Then, it is not difficult to see that for $i=1$ this difference is
minimized and equals to $\frac{n}{2}$. Analogously, the difference
between a high spine vertex $j = n - k + i$ and one of its legs is
at least:

\begin{eqnarray*}
j - (k + i \cdot \Delta) & = & n - k + i - k - i \cdot \Delta \quad \\
& = & n-2k + i(1-\Delta) \quad
\end{eqnarray*}

In this case, the difference is minimized for the largest possible
$i$, that is, $i = k$, and using the fact that $k=\frac{n}{2(\Delta
+ 1)}$, the difference is again $\frac{n}{2}$. Hence, there exists a
labeling for which the maximum difference is $\frac{n}{2}$.

\begin{figure}[t]
    \centering
    \begin{minipage}[b]{.48\textwidth}
        \centering
        \subfloat[\label{sfg:regcat-even}{Even number of spine vertices.}]
        {\includegraphics[width=\textwidth,page=1]{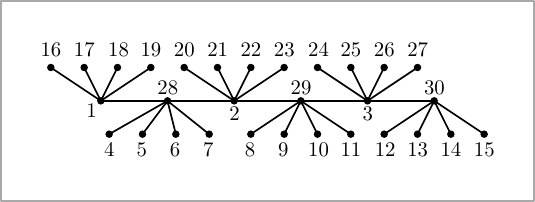}}
    \end{minipage}
    \begin{minipage}[b]{.48\textwidth}
        \centering
        \subfloat[\label{sfg:regcat-odd}{Odd number of spine vertices.}]
        {\includegraphics[width=\textwidth,page=2]{regular-caterpillars}}
    \end{minipage}
    \caption{Our optimal labeling for regular caterpillars.}
    \label{fig:regcat}
\end{figure}

Now, consider the case where the number of spine vertices is odd,
say $s = 2k+1$ for some $k \geq 1$; see Fig.~\ref{sfg:regcat-odd}.
We follow the scheme as above assigning the lowest $k+1$ numbers and
the highest $k$ numbers in an alternating fashion to the spine
vertices. The differences between adjacent spine vertices are $n -
k$ and $n - k - 1$ which is at least:

\begin{eqnarray*}
n-k-1 & = & 2k\Delta + \Delta + k  \quad \\
& \ge & k\Delta + 1 +k \quad \\
& \ge & \lceil \frac{n -\Delta}{2} \rceil \quad
\end{eqnarray*}

Let $L_s$ denote the spine vertices with labels $\le k+1$ and $H_s$
be the spine vertices with labels $> n-k$. As before, we divide the
middle range $[k+2,n-k]$ into two ranges $L_{\ell} =
[k+2,\lceil\frac{n - \Delta}{2}\rceil]$ and $H_{\ell} =
[\lceil\frac{n-\Delta}{2}\rceil  + 1,n-k]$ and label the legs of
$H_s$ with number from $L_{\ell}$ and the legs of $L_s$ with numbers
from $H_{\ell}$. For a spine vertex from $L_s$ with label $1\le i
\le k+1$, we label its $\Delta$ legs with numbers from the interval
$[\lceil\frac{n-\Delta}{2}\rceil + ((i-1) \cdot \Delta) + 1,
\lceil\frac{n-\Delta}{2}\rceil + i \cdot \Delta ]$. For a spine
vertex from $H_s$ with label $j = n - k + i$ between $n-k + 1$ and
$n$, we label its $\Delta$ legs with numbers from the interval $[k
+((i-1) \cdot \Delta) + 2,k + i \cdot \Delta +1]$.

The difference between a low spine vertex $i$ and its legs is at
least $\lceil\frac{n-\Delta}{2}\rceil + ((i-1) \cdot \Delta) + 1 -
i$ which is minimal for $i=1$, namely
$\lceil\frac{n-\Delta}{2}\rceil$ as above. The difference between a
high spine vertex $j = n - k + i$ and its legs is at least:

\begin{eqnarray*}
j - (k + i \cdot \Delta +1) & = & n - k + i- (k + i \cdot \Delta +1)  \quad \\
& \ge & n-2k+i(1-\Delta)-1 \quad
\end{eqnarray*}

In this case, the difference is minimized for $i$ as large as
possible, that is, $i = k$. Using the fact that $k = \frac{n -
(\Delta + 1)}{2(\Delta + 1)}$ and as $\Delta \ge 1$, we have that
the difference for $i = k$ is:

\begin{eqnarray*}
n-2k +k(1-\Delta)-1 & = & n-k -k\Delta -1 =n-k(1+\Delta) - 1  \quad \\
& \geq & \lceil \frac{n -\Delta}{2} \rceil \quad
\end{eqnarray*}

From the above, it follows that for a regular caterpillar with even
number of spine vertices our labeling method achieves difference
$\frac{n}{2}$ and for odd number of spine vertices it achieves
difference $\lceil \frac{n-\Delta}{2} \rceil$. Both of these are
optimal by Theorem~\ref{thm:opt-reg-cat}. This is summarized in the
following theorem.

\begin{theorem} \label{thm:reg-cat}
Let $G$ be a regular caterpillar with $n$ vertices. There exists an
optimal labeling of $G$ with value $\frac{n}{2}$
when $G$ has an even number of spine vertices, and with value
$\lceil \frac{n-\Delta}{2} \rceil$ otherwise.
\end{theorem}

\subsection{Labeling spiders with path lengths all even or all odd}
\label{sec:evenoddspider}

Let $G$ be a $n$-vertex spider consisting of $p$ paths. Recall that
by $N_l$ we denote the number of vertices at level $l$. For each $q
\in [1, N_l]$, let $v_{l, q}$ be the level-$l$ vertex that belongs
to the $q$-th path out of the paths containing level-$l$ vertices.

\begin{theorem}
Let $G$ be a spider graph with $N_e$ even level vertices. If the
paths of $G$ are all of odd length or all of even length, there
exist optimal labeling for $G$ with value $N_e+1$ and $N_e$,
respectively. \label{thm:spider}
\end{theorem}
\begin{proof}
We consider the cases where all $p$ paths of $G$ are either of odd
length or of even length separately. We first consider the case
where all $p$ paths have even length; see
Fig.~\ref{sfg:spider-alleven}. We label the center vertex as $1$.
The $N_{e}$ even level vertices will be labeled with numbers from
interval $I_e = [2,N_{e}+1]$ in increasing order of levels, i.e.,
starting with the level-$2$ vertices, followed by level-$4$
vertices, etc. The $N_{o}$ odd level vertices will be labeled with
numbers from the interval $I_o = [N_{e}+2, n]$ in the same way,
i.e., starting with the level-$1$ vertices, followed by the
level-$3$ vertices, etc. For each level, we order the vertices in
decreasing order of the lengths of the paths they belong to. More
specifically, we initially order the $p$ paths in decreasing order
of their lengths. Then, the exact label of each vertex is determined
as follow; see Fig.~\ref{sfg:spider-one-even}:

\begin{figure}[t]
    \centering
    \begin{minipage}[b]{.48\textwidth}
        \centering
        \subfloat[\label{sfg:spider-alleven}{All paths of even lengths.}]
        {\includegraphics[width=\textwidth,page=1]{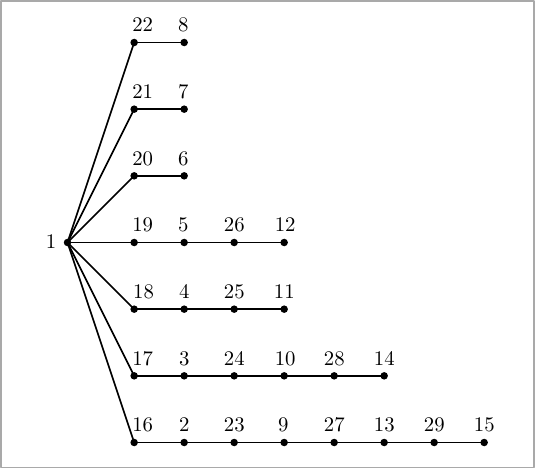}}
    \end{minipage}
    \begin{minipage}[b]{.48\textwidth}
        \centering
        \subfloat[\label{sfg:spider-allodd}{All paths of odd lengths.}]
        {\includegraphics[width=\textwidth,page=2]{even-odd-spiders}}
    \end{minipage}
    \caption{Special cases of spider graphs where our labeling is optimal.}
    \label{fig:even-odd-spider}
\end{figure}

\begin{enumerate}[-]
  \item A vertex $v_{l,q}$ belonging to an odd level $l=2i+1$, $i \ge 1$ is labeled as $N_{e}+ 1 + \sum_{k=0}^{i-1}N_{2k+1} + q$.
  \item A vertex $v_{l,q}$ belonging to an even level $l = 2i$, $i \ge 1$ is labeled as $1 + \sum_{k=1}^{i-1}N_{2k} + q$.
\end{enumerate}

We now show that the above labeling has maximum differential value
$N_{e}$. First, consider the difference between the center and a
level-$1$ vertex $v_{1,q}$. The difference is $N_{e}+1 + q - 1$. As
$q \in [1,p]$, this is at least $N_{e}+1 \ge N_{e}$.

Now, consider the difference between a vertex of level $l = 2i$ and
a vertex of level $l+1 = 2i+1$, for some $i \ge 1$. Since $N_{2k+1}
\ge N_{2k+2}$, the difference is:

\begin{eqnarray*}
v_{l+1,q} - v_{l, q} & = & N_{e}+1 + \sum_{k=0}^{i-1}N_{2k+1} + q - (1 + \sum_{k=1}^{i-1}N_{2k} + q) \quad \\
& = & n-k(1+\Delta) - 1 \quad \\
& = & N_{e} + N_{2(i-1) +1} + \sum_{k=0}^{i-2}(N_{2k+1} - N_{2k+2}) \quad\\
& \geq & N_{e} + N_{2(i-1) +1} \ge N_e \quad\\
\end{eqnarray*}

Now, consider the difference between a vertex $v_{l,q}$ at level $l
= 2i$ and a vertex $v_{l-1,q}$ on the same path. In this case and
since $N_{2k+1} \ge N_{2k+2}$, the difference is:

\begin{eqnarray*}
v_{l-1,q} - v_{l,q} & = & N_{e}+ 1 + \sum_{k=0}^{i-2}N_{2k+1} + q - (1 + \sum_{k=1}^{i-1}N_{2k} + q) \quad \\
& = & N_{e}+ \sum_{k=0}^{i-2}(N_{2k+1} - N_{2k+2}) \geq N_e \quad
\end{eqnarray*}

From the above, it follows that our labeling for the case, where all
$p$ paths have even length, has maximum differential value $N_{e}$,
as desired. Since all the paths have even length $N_{e} = N_{o} =
\lceil\frac{n}{2}\rceil$.

We now consider the case where all of the $p$ paths have odd length;
see Fig.~\ref{sfg:spider-allodd}. Let $k$ be the length of the
longest path. Since all paths have odd length, we have $N_{2i} =
N_{2i+1} \forall i \in [1,\lfloor\frac{k}{2}\rfloor]$. Also, $N_{o}
-N_{e} = N_1 + \sum_{i=1}^{\lfloor\frac{k}{2}\rfloor}(N_{2i+1} -
N_{2i}) = p$, and, $n = N_{o} + N_{e} + 1 = 2 N_{e} + p + 1$. So,
$n$ and $p$ have different parity. On the other hand,
$\lceil\frac{n-p}{2}\rceil = N_{e} +1$. We label the center vertex
as $\lceil\frac{n}{2}\rceil$. We next order the paths as
$P_{\lfloor\frac{p}{2} +1\rfloor},P_1,P_{\lfloor\frac{p}{2}
+2\rfloor},P_2 \dots$ in the decreasing order of their lengths.
Then, the exact label of each vertex is determined as follow; see
Fig.~\ref{sfg:spider-all-odd}:

\begin{enumerate}[-]
  \item A vertex $v_{l,q}$ belonging to an odd level $l=2i-1$ and $q \le \lfloor\frac{p}{2}\rfloor$ is labeled as $\sum_{m=1}^{i-1}\lfloor\frac{N_{2m-1}}{2}\rfloor + q$.
  \item A vertex $v_{l,q'}$ belonging to an odd level $l = 2i - 1$  and $q'>\lfloor\frac{p}{2}\rfloor$, $q' = q + \lfloor\frac{p}{2}\rfloor$ is labeled as  $n - (\sum_{m=1}^{i}\lceil\frac{N_{2m-1}}{2}\rceil) + q$.
  \item A vertex $v_{l,q}$ belonging to an even level $l=2i$ and $q \le \lfloor\frac{p}{2}\rfloor$ is labeled as $\lceil\frac{n}{2}\rceil+\sum_{m=1}^{i-1}\lfloor\frac{N_{2m}}{2}\rfloor + q$.
  \item A vertex $v_{l,q'}$ belonging to an even level $l = 2i$ and $q'> \lfloor\frac{p}{2}\rfloor$, $q' = q + \lfloor\frac{p}{2}\rfloor$ is labeled as $\lceil\frac{n}{2}\rceil -(\sum_{m=1}^{i}\lceil\frac{N_{2m}}{2}\rceil) +q - 1$.
\end{enumerate}

\begin{figure}[t]
    \centering
    \begin{minipage}[b]{\textwidth}
        \centering
        \subfloat[\label{sfg:spider-one-even}{The labeling scheme for spiders with all paths of even length.}]
        {\includegraphics[width=\textwidth,page=3]{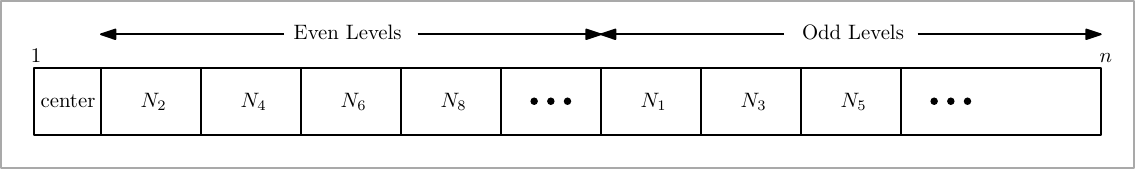}}
    \end{minipage}
    \begin{minipage}[b]{\textwidth}
        \centering
        \subfloat[\label{sfg:spider-all-odd}{The labeling scheme for spiders with all paths of odd length.}]
        {\includegraphics[width=\textwidth,page=4]{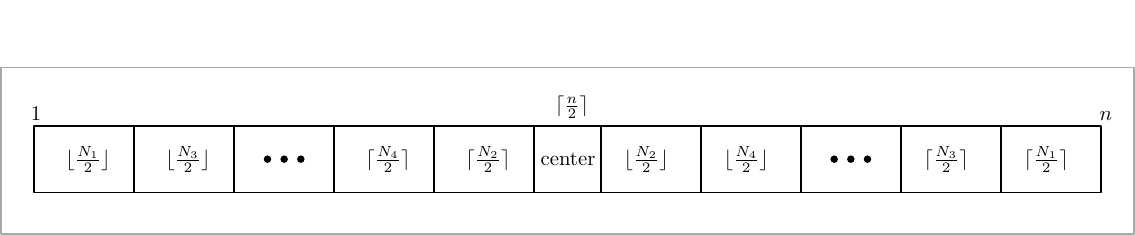}}
    \end{minipage}
    \caption{The two cases of our spider labeling scheme.}
    \label{fig:spider-numline}
\end{figure}

We now show that the above labeling has maximum differential value
$N_{e} + 1$. First, consider the difference between the center and a
level-$1$ vertex $v_{1,q}$. When $q \le \lfloor\frac{p}{2}\rfloor$,
the difference is $ \lceil\frac{n}{2}\rceil -q $, which is at least
$\lceil\frac{n}{2}\rceil - \lfloor\frac{p}{2}\rfloor =
\lceil\frac{n-p}{2}\rceil$ . When $q' >\lfloor\frac{p}{2}\rfloor ,
q' = q + \lfloor\frac{p}{2}\rfloor $, the difference is $n -
\lceil\frac{N_1}{2}\rceil + q - \lceil\frac{n}{2}\rceil$, which is
at least  $n - \lceil\frac{p}{2}\rceil + 1 - \lceil\frac{n}{2}\rceil
= \lfloor\frac{n}{2}\rfloor - \lceil\frac{p}{2}\rceil + 1 =
\lceil\frac{n-p}{2}\rceil$.

Now, consider the difference between the labels of vertices in level
$l = 2i$ and level $l-1 = 2i-1$. Since for every $m$ it holds that
$\lfloor\frac{N_{2m}}{2}\rfloor \ge
\lfloor\frac{N_{2m-1}}{2}\rfloor$, for $q \le
\lfloor\frac{p}{2}\rfloor$ the difference is:

\begin{eqnarray*}
v_{l-1,q} - v_{l,q} & = & \lceil\frac{n}{2}\rceil+ \sum_{m=1}^{i-1}\lfloor\frac{N_{2m}}{2}\rfloor + q - (\sum_{m=1}^{i-1}\lfloor\frac{N_{2m-1}}{2}\rfloor) + q) \quad \\
& = & N_{e}+ \sum_{k=0}^{i-2}(N_{2k+1} - N_{2k+2}) \quad \\
& \ge & \lceil\frac{n}{2}\rceil \quad
\end{eqnarray*}

Analogously, for $q' = q + \lfloor\frac{p}{2}\rfloor $ and $q'
>\lfloor\frac{p}{2}\rfloor$  the difference is:

\begin{eqnarray*}
v_{l-1,q'} - v_{l,q'} & = & n - (\sum_{m=1}^{i}\lceil\frac{N_{2m-1}}{2}\rceil) + q - (\lceil\frac{n}{2}\rceil - (\sum_{m=1}^{i}\lceil\frac{N_{2k}}{2}\rceil) +q - 1) \quad \\
& \ge & \lfloor\frac{n}{2}\rfloor + 1 \quad
\end{eqnarray*}

Since $ \lfloor\frac{n}{2}\rfloor + 1 \ge \lceil\frac{n}{2}\rceil$,
in both cases the difference is at least $\lceil\frac{n}{2}\rceil$.
Now, consider the difference between a vertex $v_{l,q}$ at level $l
= 2i$ and a vertex $v_{l+1,q}$ on the same path. For $q \le
\lfloor\frac{p}{2}\rfloor$ the difference is:

\begin{eqnarray*}
v_{l+1,q} - v_{l,q} & = & \lceil\frac{n}{2}\rceil+\sum_{m=1}^{i-1}\lfloor\frac{N_{2m}}{2}\rfloor + q -(\sum_{m=1}^{i}\lfloor\frac{N_{2m-1}}{2}\rfloor) + q) \quad \\
& \ge & \lceil\frac{n}{2}\rceil - \lfloor\frac{N_{2i-1}}{2}\rfloor \quad\\
& \ge & \lceil\frac{n}{2}\rceil - \lfloor\frac{p}{2}\rfloor = \lceil\frac{n-p}{2}\rceil \quad
\end{eqnarray*}

Analogously, for $q' = q + \lfloor\frac{p}{2}\rfloor $ and $q'
>\lfloor\frac{p}{2}\rfloor$ the difference is:

\begin{eqnarray*}
v_{l-1,q'} - v_{l,q'} & = & ( n - (\sum_{m=1}^{i+1}\lceil\frac{N_{2m-1}}{2}\rceil) + q) - (\lceil\frac{n}{2}\rceil - (\sum_{m=1}^{i}\lceil\frac{N_{2m}}{2}\rceil) +q - 1)\quad \\
& = & \lfloor\frac{n}{2}\rfloor  + 1 - \lceil\frac{N_{2i+1}}{2}\rceil \quad\\
& \ge & \lfloor\frac{n}{2}\rfloor  + 1 - \lceil\frac{p}{2}\rceil  = \lceil\frac{n-p}{2}\rceil
\end{eqnarray*}

We now argue that we achieve an optimal labeling for $G$ if all
paths are of even length. Recall that in this case our labeling
scheme achieves a maximum differential value of $N_e$. As all paths
have even lengths, $N_e = N_o$. Thus, $n = N_e + N_o + 1 = 2N_e +
1$. By Property~\ref{prp:upperbd}, it follows that $\dc(G) \le \lceil n
/2 \rceil = \lceil (2N_e + 1) \rceil = N_e$. In the case where $G$
is a spider with all paths of odd length, our labeling achieves a
maximum differential value of $N_e +1$, which is optimal by
Theorem~\ref{thm:opt-spider}. This completes the proof of
Theorem~\ref{thm:spider}.
\end{proof}

Now, recall that a $k$-radius star $G$ is a spider where all paths
have length exactly $k$. As $k$ is either an even or an odd number,
either all paths of $G$ are of even length or all paths are of odd
length. Thus, our labeling scheme is optimal for $k$-radius star
graphs. Hence, we can state the following as a corollary of
Theorem~\ref{thm:spider}.

\begin{corollary}\label{cor:kstar}
There exists a linear-time algorithm that computes an optimal
labeling for all radius-$k$ star graphs.
\end{corollary}

\section{Labeling General Caterpillars}
\label{sec:gen-cat}

We start with a labeling scheme for the more intuitive --but
slightly restricted-- case, where $G$ is a caterpillar and each
spine vertex has at least one leg. Then, we adapt the proposed
scheme to general caterpillars.

\subsection{Labeling for caterpillars with no missing legs}
\label{sec:catwithlegs}

We describe a labeling algorithm for $G$ that achieves differential
value at least $n/2 -\Delta -1$, comprising of two phases; the {\em
marking phase} and the {\em labeling phase}. The marking phase
places the vertices of $G$ into one of the following sets: $L_s$,
$H_s$, $M$, $L_{\ell}$, $H_{\ell}$, $L_{M_{\ell}}$ and
$H_{M_{\ell}}$; see Fig.~\ref{fig:gencat}. The labeling phase
assigns actual values to vertices of $G$.

\begin{figure}[t]
    \centering
    \includegraphics[width=\textwidth,page=1]{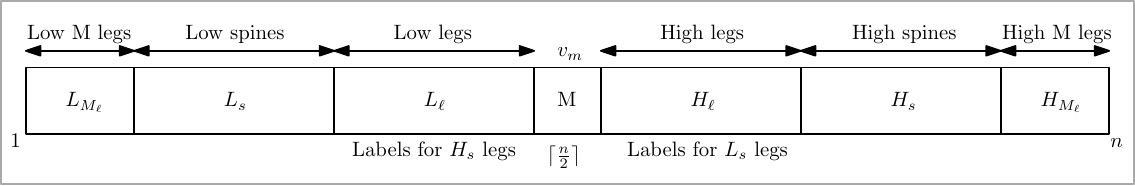}
    \caption{Allocation of labels to the vertices of a caterpillar graph; legs of $L_s$ are labeled with numbers from the interval $H_{\ell}$ and legs of $H_s$ are labeled with numbers values from the interval $L_{\ell}$.}
    \label{fig:gencat}
\end{figure}

\begin{theorem}
Let $G$ be a caterpillar with $n$ vertices in which each spine
vertex has at least one leg. There exists a labeling of $G$ with
differential coloring value at least $\lceil \frac{n}{2} \rceil -
\Delta -1$.
\label{thm:gen-cat-legs}
\end{theorem}
\begin{proof}
As already stated, our labeling algorithm is comprising of two
phases; the marking phase and the labeling phase.

\noindent{\bf Marking:} Vertices in $L_s, H_s, M$ are spine vertices
and those in $L_{\ell}$, $H_{\ell}$, $L_{M_{\ell}} \cup
H_{M_{\ell}}$ are the legs of spine vertices in $H_s$, $L_s$ and
$M$, respectively. More precisely, the vertex in $M$ is labeled as
$\lceil n/2 \rceil$, the vertices in $L_s, L_{\ell}, L_{M_{\ell}} $
are labeled with \emph{low values} $< \lceil n/2 \rceil$, and those
in $H_s, H_{\ell}, H_{M_{\ell}}$ are labeled \emph{high values} $>
\lceil n/2 \rceil$.

Start by placing all odd-numbered spine vertices in set $L_s$ and
all even-numbered ones in $H_s$, assuming that the spine vertices
are numbered according to their position in the spine path. The legs
of a spine vertex $v \in L_s$ are placed into $H_{\ell}$ and the
legs of a spine vertex $v\in H_s$ into $L_{\ell}$. Now, select one
vertex to place into set $M$ by traversing the spine vertices $S =
L_s \cup H_s$ from right to left. At each vertex $v_i \in S$ we
temporarily ignore $v_i$ and its legs from their current sets and
check if the following {\em balance condition} holds:

\begin{eqnarray}
|L_s| + |L_{\ell}| < n/2 & \mbox{ and } &  |H_s| + |H_{\ell}| \le n/2
\label{eq:balance}
\end{eqnarray}

Intuitively, when the balance condition holds the number of vertices
which we will label with low and high values are both less than
$n/2$. If the balance condition holds at $v_i \in S$, we place $v_i$
in set $M$. Otherwise, we {\em flip} $v_i$ and its legs as follows:
if $v_i$ is in set $L_s$, then we move it into set $H_s$ and move
its legs into $L_{\ell}$; else if $v_i$ is in set $H_s$, then we
move it into set $L_s$ and its legs into $H_{\ell}$.

We claim that the process always stops with a configuration where
the balance condition holds. Suppose without loss of generality that
initially $|L_s| + |L_{\ell}| < n/2$ but $|H_s| + |H_{\ell}| > n/2$.
As vertices and its legs are flipped during the traversal, if the
balance condition is not met in the end, then we must have $|L_s| +
|L_{\ell}| > n/2$ and $|H_s| + |H_{\ell}| < n/2$. Hence, at some
point in the traversal, we switch from $|L_s| + |L_{\ell}|  < n/2$
to $|L_s| + |L_{\ell}|  > n/2 $, when we flip some vertex $v_i$ and
its legs. Ignoring this vertex $v_i$ and its legs ensures that
$|L_s| + |L_{\ell}|  < n/2$. Thus, we can place vertex $v_i$ into
$M$ and stop.

Let $v_m$ be the vertex placed into $M$. Now, we partition the at
most $\Delta$ legs of $v_m$ into sets $L_{M_{\ell}}$ and
$H_{M_{\ell}}$, such that the total low and high values are:

\begin{eqnarray}
\mbox{ Low values:} & ~~ & |L_s| + |L_{\ell}| + |L_{M_{\ell}}|  + |M| = \lceil n/2 \rceil  \label{eq:totallow}\\
\mbox{ High values:} & ~~ & |H_s| + |H_{\ell}| + |H_{M_{\ell}}| = \lfloor n/2 \rfloor \label{eq:totalhigh}
\end{eqnarray}

\noindent{\bf Labeling:} Label $v_m$ as $\lceil n/2 \rceil$. Then,
label the legs of $v_m$ in $L_{M_{\ell}}$ with values from the
interval $[1,|L_{M_{\ell}}|]$ and its legs in $H_{M_{\ell}}$ with
values from the interval $[n-|H_{M_{\ell}}|+1, n]$.  As $v_m$ has at
most $\Delta$ legs, the minimum difference value between $v_m$ and
its legs is:

\begin{equation}\label{eq:diff_m_legs}
\min\{\lceil n/2 \rceil-|L_{M_{\ell}}|, \lceil n/2 \rceil - |H_{M_{\ell}}| \} \ge n/2 - \Delta
\end{equation}

For the spine vertices $S = L_s \cup H_s$, label the vertices of
$L_s$ with numbers from the interval $I(L_s) = [|L_{M_{\ell}}|+1,
|L_{M_{\ell}}|+|L_s|]$ and the vertices of $H_s$ with numbers the
interval $I(H_s) = [n-|H_{M_{\ell}}|-|H_s|+1, n-|H_{M_{\ell}}|]$.
Start by labeling the spine neighbors of $v_m$. By the balancing
procedure, $v_m$ has at most one spine neighbor from $L_s$ and one
spine neighbor from $H_s$. The neighbor from $L_s$ (if it exists) is
labeled as $|L_{M_{\ell}}| + 1$ and the neighbor from $H_s$ (if it
exists) is labeled as $n-|H_{M_{\ell}}|$.  The difference between
$v_m$ and its spine neighbors is:

\begin{equation}\label{eq:diff_m_sp}
\min\{\lceil n/2 \rceil-|L_{M_{\ell}}| , n- |H_{M_{\ell}}| - \lceil n/2 \rceil\} \ge n/2 - \Delta -1
\end{equation}

The remaining spine vertices of $L_s$ are labeled with remaining
numbers of $I(L_s)$ in increasing order. Start with the first vertex
in $L_s$ which is left of $v_m$, then move leftward labeling
vertices from $L_s$ until reaching the leftmost vertex in $L_s$.
Now, proceed to the rightmost vertex in $L_s$ and move leftward
again until $v_m$ is reached. Label the spine vertices of $H_s$ with
the remaining numbers from the interval $I(H_s)$ in exactly the same
way, i.e., in an increasing fashion starting from the vertex in
$H_s$ to the left of $v_m$ and moving leftward. As we always
increment the spine vertices by one, the difference between a spine
vertex and its adjacent spine vertices is either:

\begin{eqnarray}
 |L_s| + |L_{\ell}| +|M| + |H_{\ell}| \mbox{ or } \nonumber    \\
 |L_s| + |L_{\ell}| +|M| + |H_{\ell}| -1 \label{eq:adjspine}
\end{eqnarray}

In both cases, the difference is at least $|L_s| + |L_{\ell}|
+|M|-1$, which by Equation~\ref{eq:totallow} is at least $n/2 -
|L_{M_{\ell}}| -1 \ge n/2 -\Delta -1$.

We now describe how the leg vertices are labeled. The labels of
$L_{\ell}$ come from interval $I(L_{\ell}) = [|L_{M_{\ell}}|+ |L_s|
+1, \lceil n/2 \rceil -1]$. The labels of $H_{\ell}$ come from the
interval $I(H_{\ell}) = [\lceil n/2 \rceil +1, \lceil n/2 \rceil+
|H_{\ell}|]$. The values of $I(H_{\ell})$ are assigned in increasing
order starting with the legs of the spine vertex of $L_s$ with the
lowest value. Thus, we first label the legs of the spine vertex
labeled $|L_{M_{\ell}}| + 1$, then the legs of $|L_{M_{\ell}}| +2$,
and so on until $|L_{M_{\ell}}| + |L_s|$. Assign the values of
$I(L_{\ell})$ in decreasing order starting with the spine vertex of
$H_s$ with the highest label. Thus, the legs of $n-|H_{M_{\ell}}|$
are labeled first, then the legs of $n-|H_{M_{\ell}}|-1$ and so on
until $n-|H_{M_{\ell}}|-|H_s|+1$.

As all spine vertices have at least one leg, the difference between
the $j$-th lowest spine vertex of $L_s$ and one of its legs is at
least the value given by Equation~\ref{eq:diff-spine-leg} (same for
$H_s$):

\begin{eqnarray}
|L_s| -j + |L_{\ell}| + |M| + j & \ge &  \lceil n/2 \rceil -|L_{M_{\ell}}| -1 ~ \mbox{ by Equation~\ref{eq:totallow}} \nonumber \\
|H_s| -j + |H_{\ell}| + |M| + j & \ge &  \lceil n/2 \rceil -|H_{M_{\ell}}| -1 ~ \mbox{ by Equation~\ref{eq:totalhigh}} \label{eq:diff-spine-leg}
\end{eqnarray}

As $|L_{M_{\ell}}|$ and $|H_{M_{\ell}}|$ are both $\le \Delta$, the
differences are at least $n/2 - \Delta -1$.
\end{proof}

\subsection{Extending to general caterpillars}

We extend the labeling scheme to general caterpillars to achieve
differential value at least $n/2 -\Delta -2$.  The main idea is to
consider spine vertices with no legs as {\em pseudo-legs} of their
neighbors, thus transforming a general caterpillar into a
caterpillar where all but one spine vertex have at least one leg or
pseudo-leg; see Fig.~\ref{fig:gencatlabels}. Observe that the
rightmost spine vertex has at least one leg as otherwise it is a leg
of the spine vertex to its left.

\begin{theorem}
Let $G$ be a caterpillar with $n$ vertices. There exists a labeling
of $G$ with differential coloring value at least $\lceil \frac{n}{2}
\rceil - \Delta -2$.
\label{thm:gen-cat}
\end{theorem}
\begin{proof}
Again, our labeling algorithm is comprising of two phases; the
marking phase and the labeling phase.

\noindent {\bf Marking}:  Select a vertex $v_m$ for set $M$ as
before. Then, traverse the spine from left to right to determine
which vertices are pseudo-legs.  Let $v$ be the current spine vertex
and $v'$ be the spine vertex to the right of $v$. If $v \ne v_m$ and
$v$ currently has no legs or pseudo-legs, then first assign $v$ to
be a pseudo-leg of $v'$ and move $v$ into the corresponding set as
follows: if $v' \in L_s$, then  move $v$ into set $H_{\ell}$, if $v'
\in H_s$, then move $v$ into set $L_{\ell}$, and if $v' \in M$, then
keep $v$ in its current set.  Observe that in the first case vertex
$v$ moves from set $H_s$ into set $H_{\ell}$, and in the second it
moves from set $L_s$ into set $L_{\ell}$. Thus, the number of low
and high values to be assigned both remain $ < n/2$ and the balance
condition in Equation~\ref{eq:balance} is still satisfied.

Now, let $v_m'$ be the right neighbor of $v_m$ on the spine. If
$v_m'$ is currently a pseudo-leg, then reassign $v_m'$ to be a
pseudo-leg of $v_m$, and if $v_m' \in L_{\ell}$, then move $v_m'$
into $L_s$ and if $v_m' \in H_{\ell}$, then move it into $H_s$. Note
that the balance condition in Equation~\ref{eq:balance} is
maintained. However, as we reassigned $v_m'$ to be a pseudo-leg of
$v_m$, this may leave one vertex, namely the spine vertex to the
right of $v_m'$ with no legs. Finally, partition the {\em real} legs
of $v_m$ into sets $L_{M_{\ell}}$ and $H_{M_{\ell}}$ as before.

\begin{figure}[t]
    \centering
    \includegraphics[width=\textwidth,page=2]{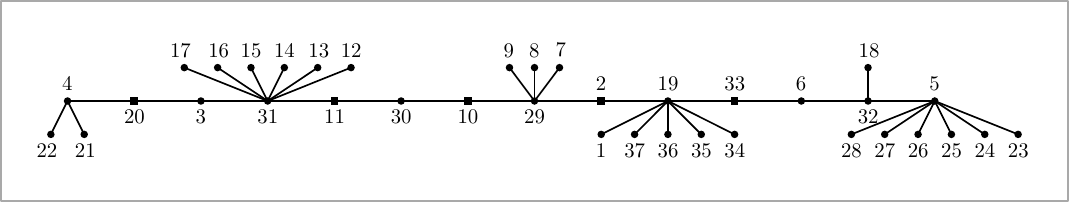}
    \caption{Labeling of a general caterpillar, where the spine vertices drawn as squares correspond to pseudo-legs.}
    \label{fig:gencatlabels}
\end{figure}

\noindent{\bf Labeling:} Label $v_m$ as $\lceil n/2 \rceil$. Then,
label $v_m$'s {\em real} legs and its neighboring vertices on the
spine also as before. Note that no the two neighboring spine
vertices of $v_m$ may be pseudo-legs of $v_m$. Still, one of them is
in the set $L_s$ and the other is in the set $H_s$, and we label
these as $|L_{M_{\ell}}| +1$ and $n-|H_{M_{\ell}}|$, respectively.
Equations~\ref{eq:diff_m_legs} and \ref{eq:diff_m_sp} still apply,
so the differential value for $v_m$ is at least $n/2 - \Delta -1$.

We label the remaining spine vertices from $S = L_s \cup H_s$ and
their legs and pseudo-legs $L_{\ell} \cup H_{\ell}$ as before. We
show that the difference between adjacent vertices is at least $n/2
- \Delta -2$. First consider two adjacent vertices from $S$. Their
difference is still given by Equation~\ref{eq:adjspine} and is at
least $|L_s| + |L_{\ell}| +|M|-1$. By Equation~\ref{eq:totallow}
this is at least $n/2 - |L_{M_{\ell}}| -1$, and as $|L_{M_{\ell}}|
\le \Delta$, the difference is $\ge n/2 -\Delta -1$.

Next, consider a vertex $v\in S$ and its legs (including its
pseudo-leg). We modify Equation~\ref{eq:diff-spine-leg} to take into
account that at most one spine vertex may have no legs. The
difference between the $j$-th lowest spine vertex of $L_s$ and one
of its legs and the $j$-th lowest spine vertex of $H_s$ and one of
its legs is at least:

\begin{eqnarray*}
|L_s| -j + |L_{\ell}| + |M| + (j-1) & \ge &  \lceil n/2 \rceil -|L_{M_{\ell}}| -1  ~ \mbox{ by Equation~\ref{eq:totallow}} \\
|H_s| -j + |H_{\ell}| + |M| + (j-1) & \ge &  \lceil n/2 \rceil -|H_{M_{\ell}}| -1 ~ \mbox{ by Equation~\ref{eq:totalhigh}}
\end{eqnarray*}

As $|L_{M_{\ell}}|$ and $|H_{M_{\ell}}|$ are both $\le \Delta$, the
difference is at least $n/2 - \Delta -1$.

Finally, consider a vertex $v\in S$ which is adjacent to $v_p$,
where $v_p$ is a pseudo-leg of $v'\in S$, $v'\ne v$. Each vertex $v
\in S$ may be adjacent to at most one such $v_p$. As $v'$ may have
$v_p$ as its only leg and as there is at most one spine vertex with
no legs, the difference depending on the label of $v$ is at least:

\begin{eqnarray*}
|L_s| -j + |L_{\ell}| + |M| + (j-2) & \ge &  \lceil n/2 \rceil -|L_{M_{\ell}}| -2  ~ \mbox{ by Equation~\ref{eq:totallow}} \\
|H_s| -j + |H_{\ell}| + |M| + (j-2) & \ge &  \lceil n/2 \rceil -|H_{M_{\ell}}| -2 ~ \mbox{ by Equation~\ref{eq:totalhigh}}
\end{eqnarray*}

As $|L_{M_{\ell}}|$ and  $|H_{M_{\ell}}|$ are $\le \Delta$, the
difference is $\ge n/2 - \Delta -2$, which completes the proof of
Theorem~\ref{thm:gen-cat}.
\end{proof}
\subsection{Comparison with the Miller-Pritikin scheme}
\label{sec:compare}

Consider a non-regular caterpillar with $s = 2k+1$ spine vertices,
where the $k+1$ odd spine vertices have one leg each and the $k$
even spine vertices have $\Delta$ legs each. This forms a bipartite
graph with disjoint vertex sets $U$ and $V$, where the $k+1$ odd
spine vertices and the $k\Delta$ legs of even spine vertices form
the set $U$, and the rest of the vertices form set $V$. The
Miller-Pritikin labeling achieves differential value equal to the
size of the smaller vertex set, i.e., $\min\{ |U|, |V|\} = 2k + 1$.
On the same graph, our labeling scheme achieves differential value
at least $n/2 -\Delta -1 = \frac{2k+1 + (k+1) + (\Delta k)}{2}
-\Delta -1 \ge \frac{k}{2}(\Delta -1)$. Let $\Delta = \Omega(n)$ and
$k=O(1)$. Then, the Miller-Pritikin scheme achieves differential
value $O(1)$, while our labeling scheme achieves differential value
$O(n)$, making it potentially worse than ours by a factor of
$\Omega(n)$.

\begin{figure}[t]
    \centering
    \begin{minipage}[b]{.48\textwidth}
        \centering
        \subfloat[\label{fig:badex}{}]
        {\includegraphics[width=\textwidth,page=3]{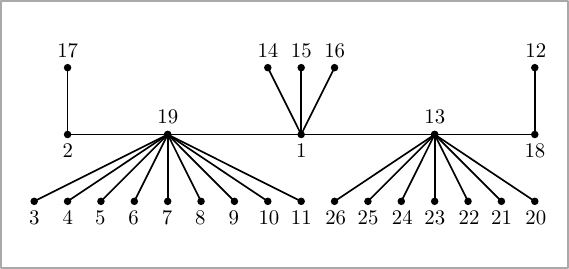}}
    \end{minipage}
    \begin{minipage}[b]{.48\textwidth}
         \centering
        \subfloat[\label{fig:badex-miller}{}]
        {\includegraphics[width=\textwidth,page=4]{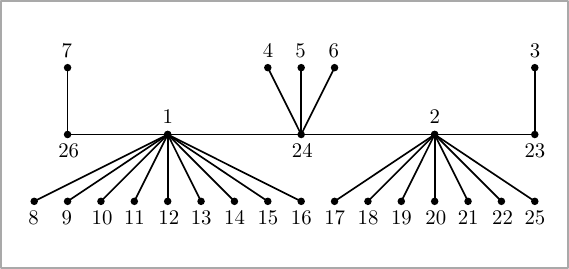}}
    \end{minipage}
    \begin{minipage}[b]{.48\textwidth}
         \centering
        \subfloat[\label{fig:badex-goodlabel}{}]
        {\includegraphics[width=\textwidth,page=5]{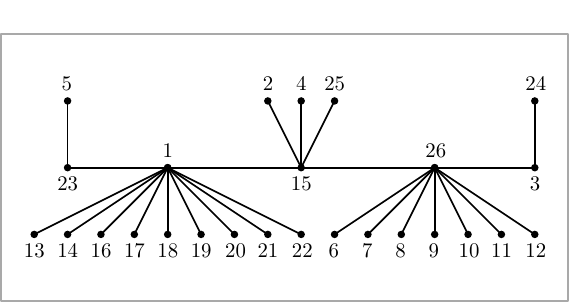}}
    \end{minipage}
    \caption{Three different labeling of a caterpillar graph.
    (a)~The labeling given by Theorem~\ref{thm:gen-cat} achieves differential value $18-13=5$.
    (b)~The Miller-Pritikin labeling achieves differential value $8-1=7$.
    (c)~A manually generated labeling achieving differential value $25-15 = 10$.}
    \label{fig:badex}
\end{figure}

Note that there are graphs for which the Miller-Pritikin scheme
achieves differential value  better than the one that our labeling
scheme achieves. Fig.~\ref{fig:badex} gives an example of a
caterpillar labeled in three different ways; by our labeling scheme
(of Theorem~\ref{thm:gen-cat}), by the Miller-Pritikin labeling
scheme and by a manually generated one. Observe that our labeling
achieves the lowest differential value; the Miller-Pritikin labeling
is only slightly better, while the manually generated labeling is
twice as good as our labeling.

\section{Differential Coloring of biconnected triangle-free outer-planar graphs}
\label{sec:outerplanar}

We first show how to obtain a $3$-equitable coloring (coloring with
$3$ colors, in which the number of vertices of each color may differ
by at most one) of a biconnected triangle-free outer-planar graph
$G$. Then, we use this coloring to obtain a differential coloring of
$G$ with value $\ge \frac{n}{3} -1$. The existence of $3$-equitable
colorability of biconnected triangle-free outer-planar graphs is
known~\cite{Wu2008985}, but not all $3$-equitable colorings can be
converted to differential colorings with the desired bound. Unlike
the existential proof in~\cite{Wu2008985}, our proof is constructive
and our algorithm also guarantees that the computed coloring can be
appropriately converted to a differential coloring of $G$ with value
$\ge \frac{n}{3} -1$.

\begin{lemma}
Given a biconnected triangle-free outer-planar graph $G$ on $n$
vertices, there is an $O(n)$-time algorithm that computes a
$3$-equitable coloring of
$G$.\label{lem:equitable-outerplanar-coloring}
\end{lemma}
\begin{proof}
Our algorithm is recursive. Consider an arbitrary edge $(u, v)$ of
$G$ that does not belong to its external face and let $f$ and $g$ be
the faces to its left and the right side, respectively, as we move
along $(u, v)$ from vertex $u$ to vertex $v$. Then, $f$ and $g$
correspond to two vertices, $v_f$ and $v_g$, of the weak dual of $G$
and $(v_f, v_g)$ is an edge in the weak dual; see
Fig.~\ref{fig:outerplanar-coloring}. The weak dual is the subgraph of the dual graph
whose vertices correspond only to the bounded faces of the primal graph.
Since the weak dual of a
biconnected outer-planar graph is a tree (not a forest), the removal
of edge $(v_f, v_g)$ results in two trees $T_f$ and $T_g$ rooted at
vertices $v_f$ and $v_g$ of the dual of $G$, respectively. For the
recursive step of our algorithm, we assume that we have already
computed a $3$-equitable coloring for the subgraph, say $G_f$, of
$G$ induced by $T_f$.

\begin{figure}[t]
    \centering
    \begin{minipage}[b]{.48\textwidth}
        \centering
        \subfloat[\label{fig:outerplanar-coloring}{}]
        {\includegraphics[width=\textwidth,page=1]{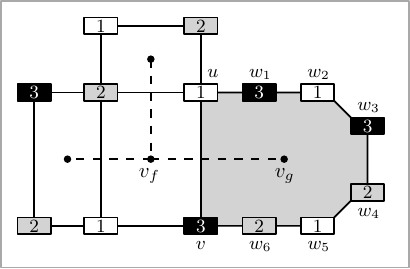}}
    \end{minipage}
    \begin{minipage}[b]{.48\textwidth}
        \centering
        \subfloat[\label{fig:outerplanar-labeling}{}]
        {\includegraphics[width=\textwidth,page=2]{outerplanar}}
    \end{minipage}
    \caption{
    (a)~A $3$-equitable coloring of a biconnected triangle-free outer-planar graph.
    (b)~The corresponding differential labeling of value $14-10=4$.}
    \label{fig:outer}
\end{figure}

We recursively compute a $3$-equitable coloring for the subgraph,
say $G_g$, of $G$ induced by $T_g$, so that the overall coloring is
a $3$-equitable coloring for $G$. Assume that $g = u, w_1, w_2,
\ldots, w_p,v$. Since $(u,v)$ is an edge of $G_f$, which is
$3$-equitable colored, vertices $u$ and $v$ have different colors;
without loss of generality let $u$ and $v$ be colored with $1$ and
$2$ from the color set $C=\{1,2,3\}$. We describe how to color the
vertices $w_1,w_2, \ldots, w_p$ along the face $g$ and maintain
$3$-equitable coloring. We consider two cases depending on the type
of equitable coloring of $G_f$.

\begin{description}
\item[Case 1:] \emph{The number of vertices of each color is the same}.
Then, if $p=1$ we can extend the $3$-equitable coloring by assigning
color $3$ to $w_1$. If $p=2$ we can easily extend the $3$-equitable
coloring by assigning colors $2$ and $1$ to $w_1$ and $w_2$. If
$p=3$ we can also extend the $3$-equitable coloring by assigning
colors $2$, $3$ and $1$ to $w_1$, $w_2$, and $w_3$.  For $p>3$, we
have the same three sub-cases modulo 3, taking into account the
colors of the first and last vertex. We note that the case where
$p=1$ as described above cannot occur as $G$ is triangle-free;
however, in faces consisting of more than three vertices, $p
\mod{3}$ can be equal to $1$.

\item[Case 2:] \emph{The number of vertices of each color is almost the same (at least one set is off by one)}.
We have six sub-cases depending on the type of color imbalance in
$G_f$ (e.g., one more vertex of color $1$, one less vertex of color
$2$, etc.). We deal with each of the sub-cases by coloring either
only $w_1$, or only $w_p$, or both $w_1$ and $w_p$, such that we
extend the $3$-equitable coloring and now the number of vertices of
each color is the same. Note that here we need the assumption that
graph is triangle-free, hence $p\ge2$. If there are uncolored
vertices of $g$, they can be colored using the case 1 coloring
strategy.
\end{description}

The rest of graph $G_g$ is processed similarly, by traversing the
dual free $T_g$ one face at a time, starting from $v_g$. This
completes the description of the recursive $3$-equitable coloring
algorithm. The algorithm begins with some face of $G$ that
corresponds to a leaf in the weak dual, for which it is easy to
compute a $3$-equitable coloring.
\end{proof}

\begin{theorem}\label{thm:biconn}
A biconnected triangle-free outer-planar graph $G$ on $n$ vertices
admits a differential coloring of value $\ge \frac{n}{3} -1$.
\end{theorem}
\begin{proof}
The proof of Lemma~\ref{lem:equitable-outerplanar-coloring} not only
implies a $3$-equitable coloring of $G$, but also suggests an order
in which the vertices of $G$ are colored. In particular, it is easy
to derive this order, if we  keep track of the time when each vertex
is colored. Let $V_i$ be the set of vertices of $G$ with color $i$,
$i=1,2,3$, and, assume without loss of generality that
$|V_1|\ge|V_2|\ge|V_3|$. We have a labeling space of $n$ slots,
which we divide into three consecutive parts, say $C_1$, $C_2$ and
$C_3$, each of which of length $|V_1|$, $|V_2|$ and $|V_3|$,
respectively. We fill up each part from left to right. Specifically,
we process the vertices of $G$ in the order in which they are
colored. Assume that we have processed zero or more vertices and let
$v$ be the next vertex in this order. If $v \in V_i$, then $v$
occupies the leftmost empty slot of $C_i$, $i=1,2,3$. Since $G$ has
a partial $3$-equitable coloring during the coloring process, the
differential coloring value will be greater than or equal to
$min(|V_1|, |V_2|,|V_3|)$, i.e., $\frac{n}{3} - 1$. 
\end{proof}

\begin{corollary}\label{cor:biconbipa}
A biconnected bipartite outer-planar graph $G$ on $n$ vertices
admits an optimum differential coloring of value equal to $\frac{n}{2} -1$.
\end{corollary}
\begin{proof}
A bipartite graph $G$ does not contain odd length cycles. Since $G$
is outer-planar, this implies that the outerface should also
consists of an even number of vertices. Hence, $n$ is even. We
compute a coloring of $G$, using the recursive algorithm described
in the proof of Lemma~\ref{lem:equitable-outerplanar-coloring}.
Since each face of $G$ has an even number of vertices, each face
that is being colored contains an even number of vertices that are
uncolored (this trivially covers the base of the recursive
algorithm), which implies that just two colors suffice to equitably
color all of its uncolored vertices (as its two already colored
vertices are unavoidably of different colors). Hence, $G$ is
$2$-equitably colorable and using an argument similar to one of
Lemma~\ref{thm:biconn}, we can prove that its differential chromatic
number is $\frac{n}{2} -1$.
\end{proof}

\section{Conclusion and Future Work}
\label{sec:conclusion}

In this paper, we proved that the differential coloring problem is
\NPH for planar graphs and we presented tight upper bounds for
regular caterpillars and spiders and closed-form optimal labeling
schemes for regular caterpillars and spiders with path lengths all
even or all odd. We notice that in a recent manuscript, Rahaman {\em
et al.}~\cite{itchy} independently obtain a result similar to our
Theorem~\ref{thm:opt-reg-cat}, i.e., an optimal labeling scheme for
regular caterpillars. For general caterpillars and biconnected
triangle-free outer-planar graphs, we presented labeling algorithms
which produce close-to-optimal labeling. Of course, there are
several natural open problems raised by our work.

\begin{enumerate}[-]
\item For general caterpillars, it is not known whether the maximum
differential coloring problem can be solved in polynomial time or
whether it is an \NPH problem. Neither ours nor the Miller-Pritikin
labeling scheme is optimal. Our algorithm is guaranteed to be within
an additive value of $\Delta  + 2$ from the optimal labeling (as
well as from the Miller-Pritikin labeling) and there are instances
where the Miller-Pritikin labeling is worse than ours by a factor
$\Theta(n)$.

\item
The decision version of the differential coloring problem is,
given a graph $G=(V,E)$ and a positive integer $k \leq |V|$, 
 determine whether $G$ has differential chromatic
number $k$. For general graphs the problem remains \NPC even for $k = 2$.
It is still
open whether the
problem is \NPH for planar graphs for a fixed constant $k$.

\item We proved that the maximum differential coloring is NP-complete
even for planar graphs. It is worth mentioning that the
computational complexity of the problem is not known for general
trees.

\item For outer-planar graphs, the known results are even fewer. We
only coped with the case of biconnected triangle-free outer-planar
graphs, which is of course a special case of the general maximum
differential coloring problem on outer-planar graphs. It still
remains open if the problem is \NPH for outer-planar graphs. Good
approximations or heuristics are also of interest.

\item There exist several other natural open problems including finding
optimal labeling schemes, proofs of NP-hardness or good
approximations for various other classes of graphs, such as lobsters
(trees in which removing all leaves results in a caterpillar),
interval graphs, cubic graphs, regular bipartite graphs, 
$1-$planar graphs.
\end{enumerate}

%
%

\section*{Acknowledgments}
We thank Jawaherul Alam, Aparna Das, Markus Geyer, Steven Chaplick and Sergey
Pupyrev for many discussions about many different variants of the
differential coloring problem.

\bibliographystyle{abbrv}
\bibliography{refs}

\end{document}